\documentclass[10pt,a4paper,twocolumn]{narms}

\usepackage{epsfig}
\usepackage{amsmath}
\usepackage{amsfonts}
\usepackage{timesmt}
\usepackage{caption}
\usepackage{algorithmic}
\usepackage{algorithm}
\usepackage[outdir=./]{epstopdf}

\usepackage[backend=biber, sorting=none, style=ieee]{biblatex}
\usepackage[utf8]{inputenc}    % utf8 support       %!!!!!!!!!!!!!!!!
\usepackage{balance}
\newtheorem{proposition}{Proposition}

\newtheorem{remark}{Remark}

\confshortname{AVEC'18}
	
\title{Adaptive MPC for Autonomous Lane Keeping}
\author{{Monimoy Bujarbaruah \& Xiaojing Zhang \& Francesco Borrelli} \\
{\aff{Department of Mechanical Engineering}} \\
{\aff{University of California Berkeley, Berkeley, CA, USA}}
\\
\\
{\authornext{H.\ Eric Tseng}}\\
{\aff{Ford Motor Company}} \\
{\aff{Dearborn, Michigan, USA}}\\
{\aff{\normalfont\normalsize E-mail: monimoy\_bujarbaruah@berkeley.edu}}\\
{\aff{\normalfont\normalsize Topics: Predictive Control, Model Adaptation, Online Learning,  Lateral Control}}
}

\abstract{This paper proposes an Adaptive Robust Model Predictive Control strategy for lateral control in lane keeping problems, where we continuously learn an unknown, but constant steering angle offset present in the steering system. Longitudinal velocity is assumed constant. The goal is to minimize the outputs, which are distance from lane center line and the steady state heading angle error, while satisfying respective safety constraints. We do not assume perfect knowledge of the vehicle lateral dynamics model and estimate and adapt in real-time  the maximum possible bound of the steering angle offset from data using a robust Set Membership Method based approach. Our approach is even well-suited for scenarios with sharp curvatures on high speed, where obtaining a precise model bias for constrained control is difficult, but learning from data can be helpful. We ensure persistent feasibility using a switching strategy during change of lane curvature. The proposed methodology is general and can be applied to more complex vehicle dynamics problems.}

\begin{document}
	
\maketitle

\section{INTRODUCTION}\label{sec:intro}
Lane keeping in (semi-)autonomous driving is an important safety-critical problem. Although simple control design often works well, two challenges can render the problem hard: $(i)$ unknown vehicle parameters \cite{byrne1995design, netto2004lateral} and $(ii)$ satisfaction of safety constraints \cite{gray2012predictive}. 
These issues become relevant in scenarios such as highway driving on sharp curves and/or unknown friction/steering offsets.

For control of constrained, and possibly uncertain systems, Model Predictive Control (MPC) has established itself as a promising tool \cite{mayne2000constrained, morari1999model, borrelli2017predictive}. Dealing with bounded uncertainties in presence of safety constraints is well understood during MPC design and this is done by means of robustifying the constraints \cite{kothare1996robust, Goulart2006, zhang:margellos:goulart:lygeros:13, ZhangKamgGeorghiouGoulLyg_Aut16}. Therefore, the use of MPC to synthesize safe control algorithms for vehicle dynamics problems is ubiquitous \cite{carvalho2015automated, carrau2016efficient, bujarbaruahtorque, liniger2017racing}. MPC based Advanced Driver Assistance Systems have been one of the most important research directions for increasing safety and mitigating road accidents. Such frameworks for vehicle lane keeping and lateral dynamics control is presented in~\cite{katzourakis2014road, Borrelli05, falcone2008mpc, choi2016mpc}. 

However, even though all the aforementioned work tackle the issue of recursive constraint satisfaction under system uncertainties, the problem of real-time adaptation of an unknown vehicle model with guarantees of constraint satisfaction, has not been thoroughly addressed. This is indeed an important problem to look into, as a recursively improved vehicle model estimate can result in increased comfort and safety over time. Data driven frameworks for learning a vehicle model are proposed in \cite{ostafew2014learning, bojarski2016end, carrau2016efficient, liniger2017racing}, but few theoretical guarantees can be established with such approaches. To address this issue, we propose an \emph{Adaptive Robust Model Predictive Control} algorithm for autonomous lane keeping. 

Adaptive controls for unconstrained systems has been widely studied and developed \cite{sastry2011adaptive, ioannou1996robust}. In recent times, this concept of online model adaptation has been extended to MPC controller design as well \cite{tanaskovic2014adaptive, lorenzen2017adaptive}. We use a similar method for designing our steering control. The model adaptation framework used in this work for estimation of vehicle model uncertainties, is built on the work of \cite{tanaskovic2014adaptive, lorenzen2017adaptive, 2018arXiv180409790B, 2018arXiv180409831B}.

In this paper, we address the problem of lane keeping under hard constraints on road boundaries and steering angle inputs. The longitudinal velocity of the vehicle is kept constant with a low level cruise controller, and we deal with only the lateral dynamics. We consider the control design problem when there exists an offset in steering angle of the vehicle, which is not exactly known to the control designer. For simplicity we have  assumed that this offset is constant with time. However, this assumption can be relaxed with appropriate bounds on maximum rate of change of the offset \cite{2017arXiv171207548T}. At every time-step, we estimate a domain, where the actual steering offset is guaranteed to belong, which we name the \emph{Feasible Parameter Set}. This Feasible Parameter Set is refined at each time, as new measurements from the steering system are obtained, thus introducing \emph{adaptation} in our algorithm. With our MPC controller, we ensure all imposed constraints are robustly satisfied at each time for all offsets in the Feasible Parameter Set.

The contributions of this paper can thus be summarized as follows:
\begin{enumerate}
    \item We introduce recursive adaptation of an unknown steering offset in an MPC framework. We guarantee robust satisfaction of imposed operating constraints in closed loop along with model adaptation. The framework yields a convex optimization problem, which can be solved in real time. 
    
    \item We guarantee recursive feasibility \cite[Chapter~12]{borrelli2017predictive} of the proposed MPC controller on a road patch of fixed curvature. On a variable curvature road, the proposed MPC algorithm is accordingly modified to a switching strategy. 
\end{enumerate}
The paper is organized as follows: in section~\ref{sec:model_des} we describe the vehicle model used. The task of lane-keeping with operating constraints is formulated in section~\ref{sec:prob_form}. Section~\ref{sec:ampc} presents the recursive model estimation algorithm and introduces the MPC problem solved. In section~\ref{sec:simul} we show numerical simulations and comparisons, and section~\ref{sec:concl} concludes the paper, laying out directions for future work. 

%%%%%%%%%%%%%%%%%%%%%%%%%%%%%%%%%%%%%%
\section{MODEL DESCRIPTION}
\label{sec:model_des}

We consider a standard lane keeping problem. A longitudinal control is considered given, and we focus only on lateral control design. We use the coordinate system defined about the center-line of the road \cite[sec.~2.5]{rajamani2011vehicle}, parametrized by the parameter $s$, which denotes distance \emph{along} the road center-line. The state space model of the car used for path following application is hence given by,  
\begin{equation}\label{eq:car_M_brief}
\dot{x}_t = A^\textnormal{c} x_t + B_1^\textnormal{c} \delta_t + B_2^\textnormal{c} r(s),
\end{equation}
where $x \in \mathbb{R}^n= [ {{e}_{\textnormal{cg}}} ~  \dot{e}_{\textnormal{cg}} ~ \delta {\psi} ~ \delta \dot{\psi}]^\top$ \cite[eq.~(2.45)]{rajamani2011vehicle}. 
Here, $e_{\textnormal{cg}}$ is the lateral position error of the vehicle's center of gravity with respect to the lane center-line, $\delta \psi$ is the yaw angle difference between the vehicle and the road, $\delta \in \mathbb{R}^m$ is the front wheel steering angle input, and $r(s)= \mathcal{C}(s) V_x$ is the yaw rate, determined by road curvature $\mathcal{C}(s)$ and vehicle longitudinal speed $V_x$. The matrices $A^\textnormal{c}\in \mathbb{R}^{n \times n}$, $B_1^\textnormal{c} \in \mathbb{R}^{n \times m}$ and $B_2^\textnormal{c} \in \mathbb{R}^{n \times m}$ are shown in the Appendix. In this particular case, $n=4$  and $m=1$, but we will maintain symbolic notation throughout the paper to present a general formulation for systems with arbitrary number of states and inputs. 

It is well-known that the system (\ref{eq:car_M_brief}) is not open loop stable in general. Therefore, a stabilizing closed loop input can be chosen as, 
\begin{align}\label{eq:steer_angle_form}
  \delta_t & = -Kx_t + \delta_{\textnormal{ff}}(s),  
\end{align}
where the feedback matrix $K$ is chosen so that the closed loop matrix $A^\textnormal{c}-B_1^\textnormal{c}K$ is stable (i.e. all the eigenvalues have negative real parts).
If given a fixed road curvature $\mathcal{C}(s)$, the feed-forward steering command $\delta_{\textnormal{ff}}(s)$ is chosen as given by \cite[eq.~(3.12)]{rajamani2011vehicle}, 
\begin{align*}
    \delta_{\textnormal{ff}}(s)  = (l_\textnormal{f} + l_\textnormal{r})\mathcal{C}(s) + K_v V_x^2\mathcal{C}(s)+K_{[1,3]}(l_\textnormal{r}\mathcal{C}(s)+ \alpha_\textnormal{r}(s)),
\end{align*}
where $l_\textnormal{f}$ and $l_\textnormal{r}$ are longitudinal distance from vehicle center of gravity to front and rear axles respectively, $K_v$ is the under-steer gradient, $V_x$ is the vehicle longitudinal speed, and $\alpha_\textnormal{r}(s)$ is the slip angle at rear tires on the road of curvature $\mathcal{C}(s)$. Here $K_{[1,3]}$ denotes the element in first row and third column of the feedback matrix $K$ (as we have just a scalar input, $\delta_t$). The resultant steady state trajectory $x_{\textnormal{ss}}(s)$ and input $\delta=\delta_{\textnormal{ss}}(s)$ obtained are given as \cite[eq.~(3.14)]{rajamani2011vehicle},
\begin{subequations}\label{eq:ss_vals}
\begin{align}
x_{\textnormal{ss}}(s) & = [0, 0, -l_\textnormal{r}\mathcal{C}(s)+\alpha_\textnormal{r}(s), 0]^\top,\\
\delta_{\textnormal{ss}}(s) & = (l_\textnormal{f}+l_\textnormal{r})\mathcal{C}(s)+K_v V_x^2 \mathcal{C}(s).
\end{align}
\end{subequations}

%%%%%%%%%%%%%%%%%%%%%%%%%%%%%%%%%
%%%%%%%%%%%%%%%%%%%%%%%%%%%%%%%%%

\section{PROBLEM FORMULATION}
\label{sec:prob_form}

In this paper, we consider deviations from the steady state trajectory (\ref{eq:ss_vals}), given a road curvature $\mathcal{C}(s)$, and we aim to regulate such deviations (errors) while imposing constraints. Accordingly, we consider the error model from the steady state trajectory (\ref{eq:ss_vals}) as,
\begin{align}\label{eq:err_model_c}
\Delta \dot{x}_t = A^\textnormal{c}\Delta x_t + B_1^\textnormal{c}\Delta\delta_t + E^\textnormal{c} \theta_\textnormal{a},
\end{align}
where $\Delta x_t = x_t-x_{\textnormal{ss}}(s)$ and $\Delta \delta_t = \delta_t - \delta_{\textnormal{ss}}(s)$. We assume the presence of an offset in the steering system of the vehicle, that is modeled with a parameter vector $\theta_\textnormal{a} \in \mathbb{R}^p$, which enters the dynamics (\ref{eq:err_model_c}) linearly with a matrix $E^\textnormal{c} \in \mathbb{R}^{n \times p}$. We discretize system (\ref{eq:err_model_c}) using forward Euler method with sampling time $T_{\textnormal{s}}$, and obtain:
\begin{align}\label{eq:discr_sys}
    \Delta x_{t+1} & = A\Delta x_t + B_1\Delta \delta_t + E\theta_\textnormal{a} + w_t, 
\end{align}
where $A = (I+T_{\textnormal{s}}A^\textnormal{c})$, $B_1 = T_{\textnormal{s}}B_1^\textnormal{c}$, $E = T_\textnormal{s}E^\textnormal{c}$. We consider a bounded uncertainty $w_t \in \mathbb{W}$ introduced by discretization and potential process noise components, where $\mathbb{W}$ is considered convex and polytopic. We impose constraints of the form,
\begin{align}\label{eq:constraints}
    C \Delta x_t + D \Delta \delta_t \leq b,
\end{align}
which must be satisfied for all uncertainty realizations $w_t\in\mathbb{W}$. The matrices $C \in\mathbb{R}^{s\times n}$, $D \in\mathbb{R}^{s \times m}$ and $b\in\mathbb{R}^{s}$ are assumed known. The control objective is to keep $\Delta x_t$ small, while satisfying tracking error constraints in input and states given by (\ref{eq:constraints}). Our goal is to design a controller that solves the infinite horizon robust optimal control problem:
\begin{equation}\label{eq:OP_inf}
	\begin{array}{clll}
		\hspace{0cm} V^{\star}(\Delta x_S) = \\ [1ex]
	\displaystyle\min_{\Delta \delta_0,\Delta \delta_1(\cdot),\ldots} & \displaystyle\sum\limits_{t\geq0} \ell \left( \Delta \bar{x}_t, \Delta \delta_t\left(\Delta\bar{x}_t\right) \right) \\[1ex]
		\text{s.t.}  & \Delta x_{t+1} = A\Delta x_t + B_1 \Delta \delta_t(\Delta x_t) \\
	   &~~~~~~~~~~~~~~~~~~~~~~~~~+ E\theta_{\textnormal{a}} +w_t,\\[1ex]
		& C\Delta x_t + D \Delta \delta_t \leq b,\ \forall w_t \in \mathbb W,\ \\[1ex]
		&  \Delta x_0 = \Delta x_S,\ t=0,1,\ldots,
	\end{array}
\end{equation}
where $\theta_{\textnormal{a}}$ is the constant steering offset present in the steering system and $\Delta \bar{x}_t$ denotes the disturbance-free nominal state. The nominal state is propagated in time with (\ref{eq:discr_sys}) excluding the effects of offset and additive uncertainties, i.e. $\Delta \bar x_{t+1} = A\Delta \bar x_t + B_1 \Delta \bar \delta_t$. It is utilized to obtain the nominal cost, which is minimized in optimization problem (\ref{eq:OP_inf}). We point out that, as system (\ref{eq:discr_sys}) is uncertain, the optimal control problem (\ref{eq:OP_inf}) consists of finding input policies $[\Delta \delta_0,\Delta \delta_1(\cdot),\Delta \delta_2(\cdot),\ldots]$, where $\Delta \delta_t: \mathbb{R}^{n}\ni \Delta x_t \mapsto \Delta \delta_t = \Delta \delta_t(\Delta x_t)\in \mathbb{R}^m$ are state feedback policies. 

In this paper, we approximate a solution to problem (\ref{eq:OP_inf}), by solving a finite time constrained optimal control problem. Moreover, we assume that steering offset $\theta_{\textnormal{a}}$ in (\ref{eq:OP_inf}) is not known exactly. Therefore, we propose a parameter estimation framework to refine our knowledge of $\theta_{\textnormal{a}}$ and thus, improve lateral controller performance.

\begin{remark}
Instead of considering the error dynamics (\ref{eq:err_model_c}) for control design, we can also discretize (\ref{eq:car_M_brief}) and formulate a control design problem. In that case, road curvature $\mathcal{C}(s)$ appears as an exogenous ``disturbance'' input through the term $r(s)$. Hence, we must design robust control algorithms for the worst possible values of $\mathcal{C}(s)$, which results in extremely conservative control. Therefore, in our method, we forgo this conservatism to attain performance (in terms of ability to handle higher $V_x$). As shown later in this paper, this results in a switching strategy and potential constraint violations upon sudden high change of $\mathcal{C}(s)$ at high $V_x$. Therefore, one might choose either formulation, bearing in mind the aforementioned performance vs safety trade-off.
\end{remark}

\begin{remark}
It is worth noting that any uncertainty in the vehicle parameters such as tire friction coefficients and mass etc, appear as parametric uncertainties in the matrices $A$ and $B_1$ in (\ref{eq:discr_sys}). One can also upper bound effect of such uncertainties with an additive uncertain term and propagate the system dynamics (\ref{eq:discr_sys}) with a chosen set of nominal $A,B_1$ matrices. However, for the sake of ease of numerical simulations, we have focused on only an additive steering angle offset. This is done without the loss of generality of the proposed approach.  
\end{remark}

%%%%%%%%%%%%%%%%%%%%%%%%%%%%%%%

\section{ADAPTIVE MPC ALGORITHM}
\label{sec:ampc}

In this section, we present the proposed \emph{Adaptive Robust MPC} algorithm for lane keeping with constraints. We also formulate the Set Membership Method based steering offset estimation, which yields model adaptation in our framework.  

\subsection{Steering Offset Estimation}
\label{ssec:offset_adap}

We characterize the knowledge of the steering offset $\theta_{\textnormal{a}}$ by its domain $\Theta$, called the \emph{Feasible Parameter Set}, which we estimate from previous vehicle data. This is initially chosen as a polytope
$\Theta_0$. The set is then updated after each time step upon gathering input-output data. The updated \emph{Feasible Parameter Set} at time $t$, denoted by $\Theta_t$, is given by,
\begin{multline} \label{eq: fps_update}
    \Theta_t =  \{\theta_t \in \mathbb{R}^p: \Delta  x_{t}-A\Delta x_{t-1}-B_1 \Delta \delta_{t-1}-E\theta_t \in \mathbb{W}, \\~\forall t \geq 0\},
\end{multline}
where $\Delta x_t,~\forall t\geq 0$ denotes the realized trajectory in closed loop. It is clear from (\ref{eq: fps_update}) that as time goes on, new data is progressively added to improve the knowledge of $\Theta$, without discarding previous information. For any $t$, knowledge from all previous time instants is included in $\Theta_t$. Thus, updated Feasible Parameter Sets are obtained with intersection operations on polytopes, and so $\Theta_{t+1} \subseteq \Theta_t$ for all $t\geq 0$. 

\subsection{Control Policy Approximation}\label{ssec:con_pol}

We consider affine state feedback policies $\pi(.)$ of the form as in \cite[Chapter~3]{kouvaritakis2016model}
\begin{align}\label{eq:cont_policy}
\pi_t(.): \Delta\delta_t(\Delta x_t)= -K \Delta x_t + v_t.    
\end{align}
where $K\in\mathbb{R}^{m \times n}$ is the fixed stabilizing state feedback gain introduced in (\ref{eq:steer_angle_form}) and $v_t$ is an auxiliary control input. 

One might also consider affine disturbance feedback policies introduced in \cite{Goulart2006}, that enable optimization over feedback gain matrices, as this is a convex optimization problem. That is not the case for state feedback policies (\ref{eq:cont_policy}), where optimizing over $K$ is non-convex. We use (\ref{eq:cont_policy}) for the sake of simplicity. Note that the subsequent analysis \emph{stays valid} when switched to disturbance feedback policies. 

\subsection{Robust MPC Problem}
We need to ensure that constraints (\ref{eq:constraints}) are satisfied $\forall w_t \in \mathbb W$, in presence of the unknown steering offset $\theta_{\textnormal{a}}$. Therefore, (\ref{eq:constraints}) are imposed for all potential steering offsets in the Feasible Parameter Set, i.e. $\forall \theta_t \in 
\Theta_t$ and for all $t \geq 0$. So we can reformulate (\ref{eq:OP_inf}) as a tractable finite horizon robust MPC problem as,
\begin{equation}\label{eq:rob_MPC}
	\begin{array}{clll}
% 		\hspace{0cm} V_{t\to t+N}(\Delta \bar{x}_t) = \\ [1ex]
	\displaystyle \min_{v_{t|t},\ldots,v_{t+N-1|t}} &  \displaystyle\sum\limits_{k=t}^{t+N-1} \ell \left( \Delta \bar{x}_{k|t}, v_{k|t} \right)  \\[1ex]
		\text{s.t.}  & \Delta x_{k+1|t} = A_{\textnormal{cl}}\Delta x_{k|t} + B_1v_{k|t} \\
		&~~~~~~~~~~~~~~~~~~~~+ E\theta_{k|t} + w_{k|t},\\[1ex]
		& C \Delta x_{k|t} + D \Delta \delta_{k|t} \leq b, \\
		& \Delta x_{t|t}= \Delta x_{t},~ \Delta x_{t+N|t} \in\mathcal{X}^N_t,\\&~~~~~~~~~~\forall w_{k|t} \in \mathbb{W},~\forall \theta_{k|t} \in \Theta_t,\\[1ex]
		& k=t,\ldots,t+N-1,
	\end{array}
\end{equation}

where $A_{\textnormal{cl}} = (A-B_1K)$ is the stable closed loop matrix for error state $\Delta x$ dynamics after applying the control policy (\ref{eq:cont_policy}). Moreover, $\mathcal{X}^N_t = \{x: Y_tx \leq z_t,~Y_t \in \mathbb{R}^{r_t \times n},~z_t \in \mathbb{R}^{r_t}\}$ is the terminal state constraint set for the MPC problem, which is chosen to ensure recursive feasibility \cite[sec.~12.3]{borrelli2017predictive} of (\ref{eq:rob_MPC}) in closed loop for a patch of road with curvature $\mathcal{C}(s)$. The properties of the set $\mathcal{X}^N_t$ is elaborated in the following section of this paper. 

After solving the optimization problem (\ref{eq:rob_MPC}) at time $t$, we apply the closed loop control policy as,
\begin{align}\label{eq:cl_control}
\Delta \delta_t^\star (\Delta x_t) = -K \Delta x_t + v^\star_{t|t}.   
\end{align}
We then resolve (\ref{eq:rob_MPC}) and continue the process of applying only the first input in closed loop. This yields a receding horizon strategy. In Appendix, we present a formulation of (\ref{eq:rob_MPC}) so that it can be efficiently solved in real-time with existing solvers.

\subsection{Terminal Set and Recursive Feasibility}
The terminal set $\mathcal{X}^N_t$ is chosen as the robust positive invariant set \cite{Goulart2006} for the system (\ref{eq:discr_sys}) with a feedback controller $u = -Kx$ and for all $w \in \mathbb{W},~\theta \in \Theta_t$. This set has the properties that,
\begin{multline}\label{eq:term_set}
    \mathcal{X}^N_t = \{x: A_{\textnormal{cl}}x + w + E\theta \in \mathcal{X}^N_t,~\forall w \in \mathbb{W},~\forall \theta \in \Theta_t,\\
    (C-DK)x \leq b\}.
\end{multline}
In other words, once any state $x$ lies within the set $\mathcal{X}^N_t$, it continues to do so indefinitely, despite all values of feasible offsets and uncertainties, satisfying all imposed constraints. Algorithms to compute such an invariant set can be found in \cite{borrelli2017predictive, kouvaritakis2016model}.

\begin{remark}
It is to be noted that due to the property $\Theta_{t+1} \subseteq \Theta_t$, we have $\mathcal{X}^N_{t+1} \supseteq \mathcal{X}^N_t,~\forall t\geq 0$. Therefore, if an invariant terminal set is computed at any time instant $t_0$, it continues to be a valid terminal set for all $t \geq t_0$. We have nonetheless chosen to compute $\mathcal{X}^N_t$ repeatedly for all $t$, ensuring that model adaptation (\ref{eq: fps_update}) can progressively lower conservatism while solving (\ref{eq:rob_MPC}).  
\end{remark}

\begin{proposition}
With the terminal set $\mathcal{X}^N_t$ as chosen above, the MPC problem (\ref{eq:rob_MPC}) along with controller (\ref{eq:cl_control}) in closed loop is recursively feasible on a road of fixed curvature $\mathcal{C}(s)$ , that is, if it is feasible at $t=0$, it is feasible for all $t \geq 0$ along that road. 
\end{proposition}

\begin{proof}
Let us consider problem (\ref{eq:rob_MPC}) is solved successfully at time $t$. Now let us denote the optimal steering input policies obtained at time $t$ be given by $[\pi^\star_{t|t}(.),\pi^\star_{t+1|t}(.),\cdots,\pi^\star_{t+N-1|t}(.)]$. Now, consider a policy sequence at the next time instant as
\begin{align}\label{eq:feas_seq_next_DF}
    \pi_{t+1}(.) = [\pi^*_{t+1|t}(.),\pi^*_{t+2|t}(.),\cdots,\pi^*_{t+N-1|t},-Kx],
\end{align}
assuming the road curvature is unaltered during the time span. With this feasible input policy sequence, we know that (\ref{eq:constraints}) are going to be satisfied, since from the property of Feasible Parameter Sets, $\Theta_{t+1} \subseteq \Theta_{t}$. Moreover, this also guarantees terminal state satisfies $x_{t+N+1|t+1} \in \mathcal{X}^N_{t+1}$. Therefore, the sequence (\ref{eq:feas_seq_next_DF}) is a feasible input sequence for the $(t+1)^{\textnormal{th}}$ time instant. Hence the MPC algorithm is recursively feasible. This concludes the proof.
\end{proof}

Therefore, with our algorithm we can guarantee that once operating constraints are met along a road patch of curvature $\mathcal{C}(s)$, they continue to do so as long as the curvature remains the same. This is important to ensure safety, once (\ref{eq:cl_control}) is applied in closed loop to (\ref{eq:discr_sys}). The algorithm can be summarized as:

\begin{algorithm}
    \caption{
Adaptive MPC for Lane Keeping
    }
    \label{alg1}
        \begin{algorithmic}[1]
\WHILE{$\theta_a = \textnormal{constant}$}
      \STATE Obtain road curvature $\mathcal{C}(s)$. Compute corresponding steady state trajectory $x_{\textnormal{ss}}(s)$ and steering angle input $\delta_{\textnormal{ss}}(s)$. Set $t=0$; initialize Feasible Parameter Set $\Theta_0$. 
      
      \WHILE{$\mathcal{C}(s)$ unchanged w.r.t. step~2}     
      \STATE Compute the terminal invariant set $\mathcal{X}^N_t$. Compute $v_{t|t}^{\star}$ from (\ref{eq:rob_MPC}) and apply steering command $\delta_t = \delta_{\textnormal{ss}}(s)-K(x_t-x_{\textnormal{ss}}(s))+ v_{t|t}^{\star}$. 

     \STATE Update $\Theta_{t+1}$ using (\ref{eq: fps_update}). Set $t=t+1$. 
     
     \STATE Estimate $\mathcal{C}(s)$. 
     \ENDWHILE
     
     \STATE set $\Theta_0 = \Theta_{t}$. Return to step~2.
    
    \ENDWHILE
    \end{algorithmic}
\end{algorithm}

The most important caveat of Algorithm~1 is that we can only guarantee recursive feasibility of (\ref{eq:rob_MPC}), as long as the road curvature $\mathcal{C}(s)$ stays constant. When a curvature change is detected, say at time $t_f$, we reset $\Delta x_0$ and $\Delta \delta_0$ at that time, measuring from the new reference trajectory $x_{\textnormal{ss}}(s)$ and $\delta_{\textnormal{ss}}(s)$. We then start solving the robust MPC problem (\ref{eq:rob_MPC}) again, with the Feasible Parameter Set starting at $\Theta_{t_f}$. This results in a switching strategy. During transition between curve switches, constraints (\ref{eq:constraints}) are softened for feasibility of (\ref{eq:rob_MPC}), and any violation is heavily penalized. As mentioned previously in Remark~1, such constraint violations can be avoided at the cost of very highly restrictive longitudinal velocity $V_x$ limits, by formulating an optimal control problem robust to all possible curvatures $\mathcal{C}(s),~\forall s$.

%%%%%%%% REMARKS

% \begin{remark}
% Instead of considering affine state feedback policies of the form (\ref{eq:cont_policy}), one might also consider affine disturbance feedback policies as introduced in \cite{Goulart2006}. Disturbance feedback strategy also enables optimization over feedback gain matrices, as this is a convex optimization problem. That is not the case for state feedback policies, where optimizing over $K$ is non-convex. In this paper we use (\ref{eq:cont_policy}) for the sake of simplicity. The analysis stays valid when switched to disturbance feedback policies.   

% \end{remark}

\begin{remark}
For an appropriately chosen $\Theta_0$, one can also assume that the steering offset parameter $\theta_{\textnormal{a}} \in \Theta_0$ for all times. This is useful for resetting the Feasible Parameter Set to $\Theta_0$, in case any time variation in offset is suspected after prolonged operation. We can then restart the algorithm. This allows immediate adaptability of the proposed algorithm, without having to re-tune a controller from the scratch. 
\end{remark}

%%%%%%%%%%%%%%%%%%%%%%%%%%%%%%%%

\section{SIMULATION RESULTS}\label{sec:simul}
In this section we present a detailed numerical example with the proposed Adaptive Robust MPC algorithm. The vehicle parameter values chosen for the simulation are given in Table~1.

\begin{table}[ht]\label{tab:table1}
	\renewcommand{\arraystretch}{1.9}
	\centering
	 \begin{tabular}{|c |c |}
		\hline
		\bfseries Parameter & \bfseries Value \\
		\hline\hline
		$M$ & $1830$ kg\\ $J_z$ & $3477$ kgm$^2$\\
		$l_{\textnormal{f}}$ & $1.152$ m \\$l_{\textnormal{r}}$ & $1.693$ m\\
		$c_{\textnormal{f}}$ & $40703$ N/rad \\ $c_{\textnormal{r}}$ & $64495$ N/rad\\
		\hline
	\end{tabular}
	\caption{Vehicle Model Parameters}
\end{table}
We consider the task of following the center-line of a lane at longitudinal speed of $V_x = 30~\textnormal{m/s}$. The reference lane is chosen as a circular arc of curvature $\mathcal{C}(s) = 1$ meters. 

We impose constraints on the deviation from the steady state trajectories' lane center-line and heading angle error. Moreover, constraints are also imposed on the steering angle deviation from the steady state value $\delta_{\textnormal{ss}}(s)$ found for a fixed curvature. All the design parameters are elaborated in Table~2.

Also, uncertainty in (\ref{eq:discr_sys}) $w_t \in \mathbb W = \{w \in \mathbb R^4: ||w||_{\infty} \leq 0.5 \}$. The initial Feasible Parameter Set is defined as 
\begin{align}\label{eq:init_FPS_sim}
    \Theta^{(0)} = \{ \theta \in \mathbb{R}^2: \begin{bmatrix}-0.2 \\ -0.3 \\
	\end{bmatrix} \leq \theta \leq \begin{bmatrix}0.2 \\ 0.3
	\end{bmatrix}\}.  
\end{align}
The matrix $E \in \mathbb{R}^{4 \times 2}$ is picked as a matrix of ones. The feedback gain $K$ in \eqref{eq:cont_policy} is chosen as optimal LQR gain for system \eqref{eq:discr_sys} with parameters $Q$ and $R$ defined in Table~2. 
The linear programs arising in the optimization problem are solved with GUROBI solver \cite{gurobi2015gurobi} in MATLAB. 
\begin{table}[t]\label{tab:table2}
	\renewcommand{\arraystretch}{1.2}
	\centering
	 \begin{tabular}{|c |c|}
		\hline
		\bfseries Parameter & \bfseries Value \\
		\hline\hline
		$p$ & $2$\\
	   $\Delta e_{\textnormal{cg}}$ & $[-4,4]$~m\\
	  $\Delta \delta \psi$ & $[-0.7,0.7]$~rad\\
	  $\Delta \delta$ & $[-\frac{\pi}{3},\frac{\pi}{3}]$~rad\\
	  $\theta_{\textnormal{a}}$&$[-0.17,0.26]^\top$\\
	  $T_\textnormal{s}$ & $0.1$\\
		$Q$ & $\textbf{diag}(2,2,2,2)$ \\ $R$ & $\textbf{diag}(1)$ \\
		$N$&$6$\\ \hline
	\end{tabular}
	\caption{MPC Design Parameters}
\end{table}
We illustrate two key aspects of the algorithm, namely $(i)$ recursive constraint satisfaction despite model mismatch, and $(ii)$ ease of reset in case of slight variation in steering offset, avoiding the need to re-tune.
\subsection{Robust Constraint Satisfaction}
In Algorithm~1, uncertainty is adapted in (\ref{eq: fps_update}) using known bounds of $\mathbb{W}$. So, it is always ensured that the unknown true steering offset $\theta_{\textnormal{a}}$ always lies within the Feasible Parameter Set $\Theta_t$ for all values of time. This can be seen in Fig.~\ref{fig:fps}.
\begin{figure}[H] 
	\centering
	\includegraphics[width=0.5\textwidth]{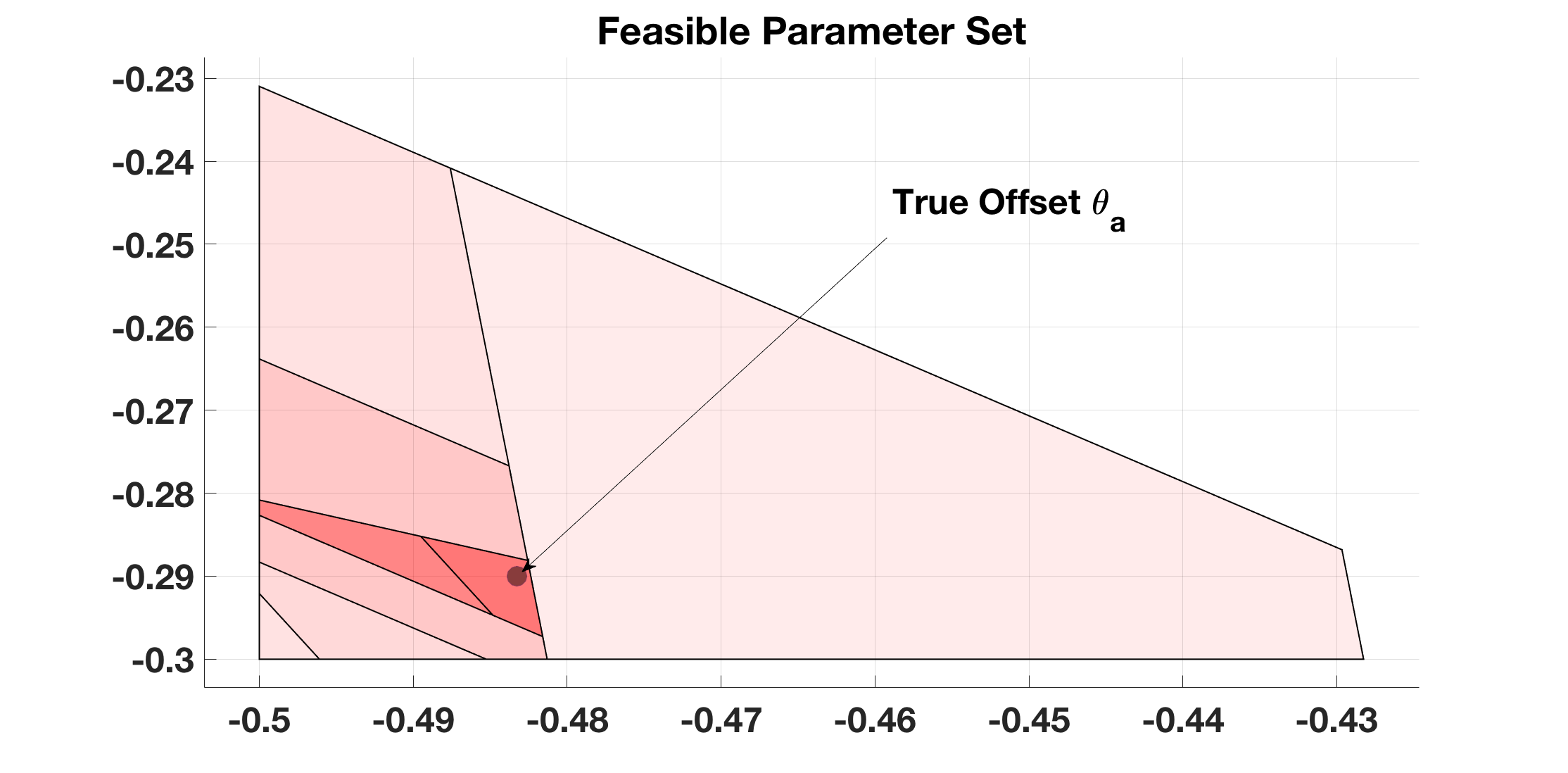}
	\caption{Feasible Parameter Set Adaptation}
	\label{fig:fps}
\end{figure}
Due to the above property, the proposed algorithm attains robustness against the unknown steering offset $\theta_{\textnormal{a}}$. We highlight this from Fig.~\ref{fig:Sconstr_bc} and Fig.~\ref{fig:Uconstr_bc}. 
\begin{figure}[ht] 
	\centering
	\includegraphics[width=0.5\textwidth]{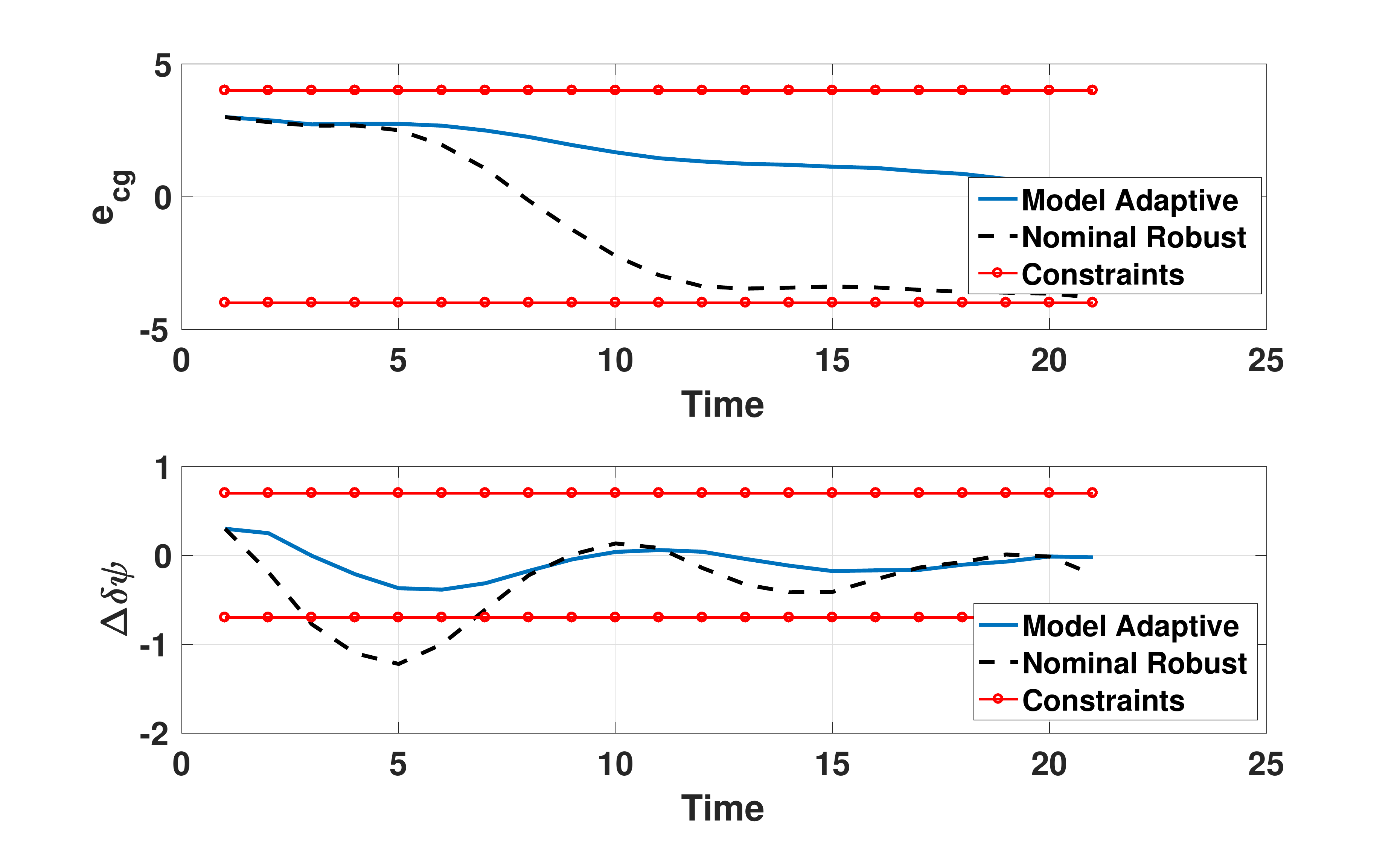}
	\caption{Adaptive vs Standard Robust MPC States}
	\label{fig:Sconstr_bc}
\end{figure}
\begin{figure}[H] 
	\centering
	\includegraphics[width=0.5\textwidth]{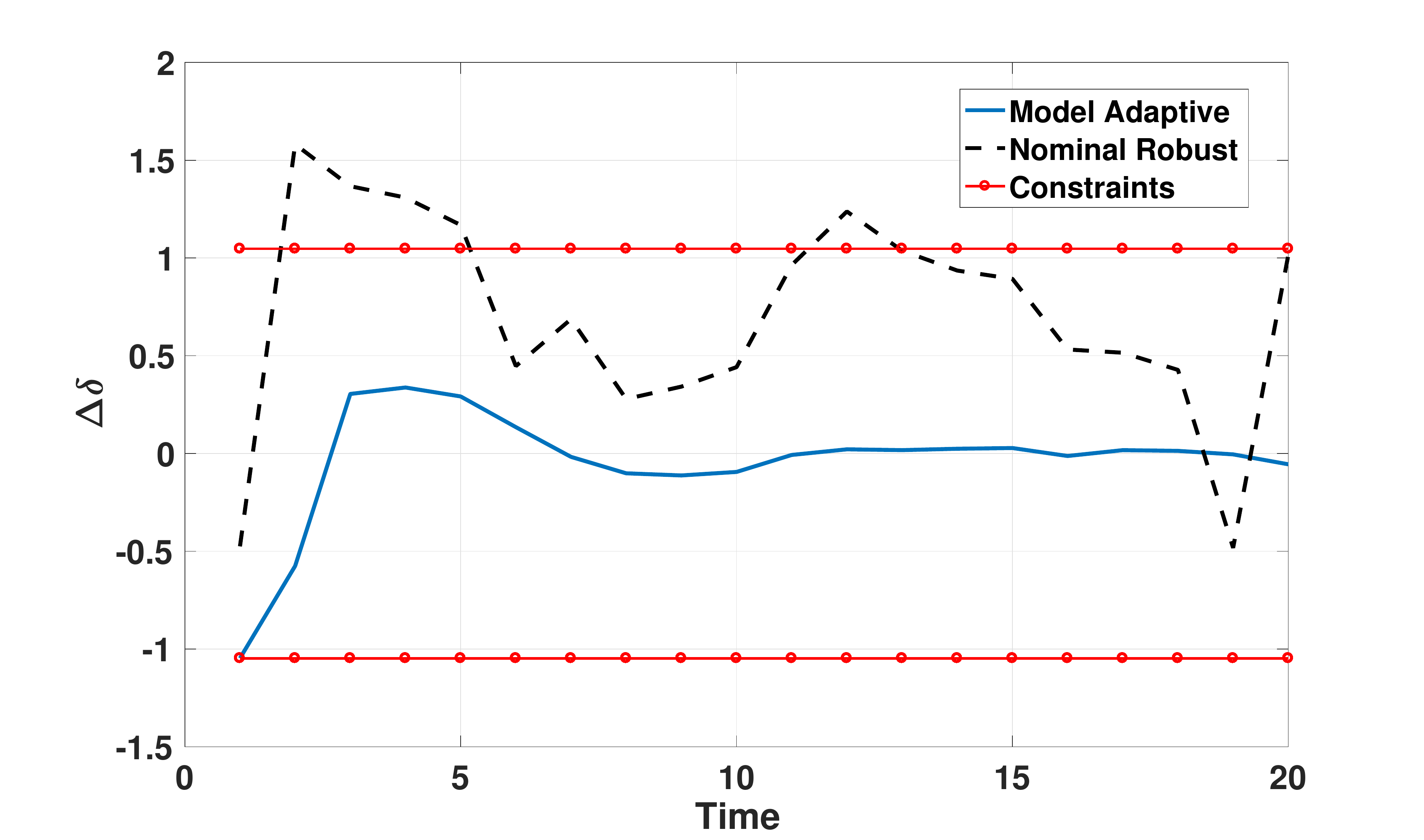}
	\caption{Adaptive vs Standard Robust MPC Control}
	\label{fig:Uconstr_bc}
\end{figure}
Here we compare our algorithm with a standard MPC formulated with a ``nominal" vehicle model that is robust against the disturbance $w \in \mathbb{W}$, but not against model biases. It has a wrong estimate of steering offset. Due to such model mismatch, the standard ``Nominal Robust'' MPC yields significant constraint violations. Contrarily, our algorithm is cautious and considers model uncertainty, thus always satisfies the constraints.

\subsection{Ease of Reset}
Now we demonstrate our algorithm's ability to avoid rigorous re-tuning, unlike a standard LQR controller. We consider a case when the steering offset value is slightly increased, after prolonged operation of the vehicle. We do not modify the initial Feasible Parameter Set $\Theta_0$ and the weights in our algorithm. Comparison shown in Fig.~\ref{fig:reset_c_ampc} and Fig.~\ref{fig:reset_c_lqr} highlight that unlike our controller, the LQR controller does not satisfy the constraints, and demands re-tuning as the offset is altered.

\begin{figure}[H] 
	\centering
	\includegraphics[width=0.5\textwidth]{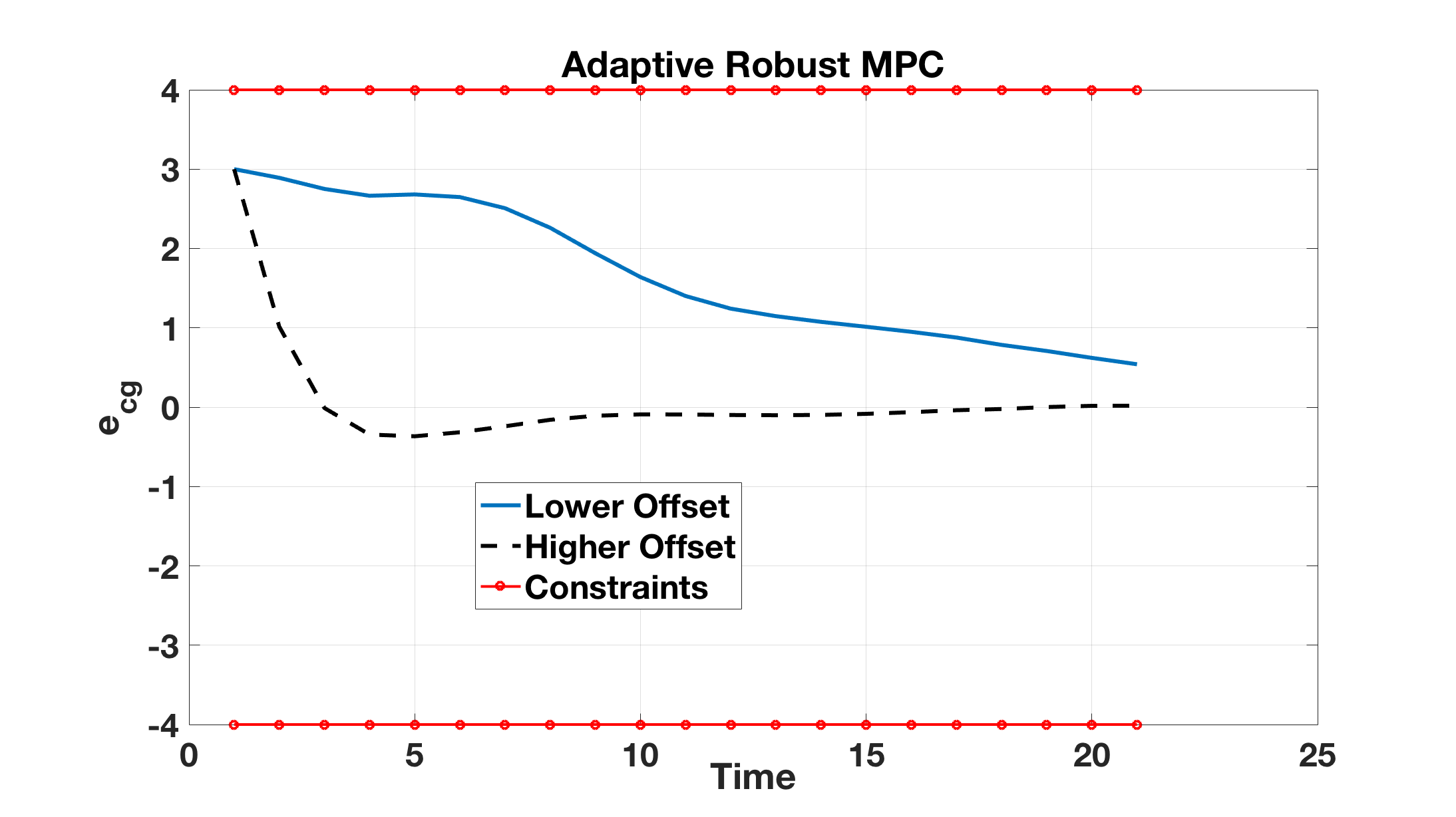}
	\caption{Adaptive MPC for Varying Offset}
	\label{fig:reset_c_ampc}
\end{figure}

\begin{figure}[H] 
	\centering
	\includegraphics[width=0.5\textwidth]{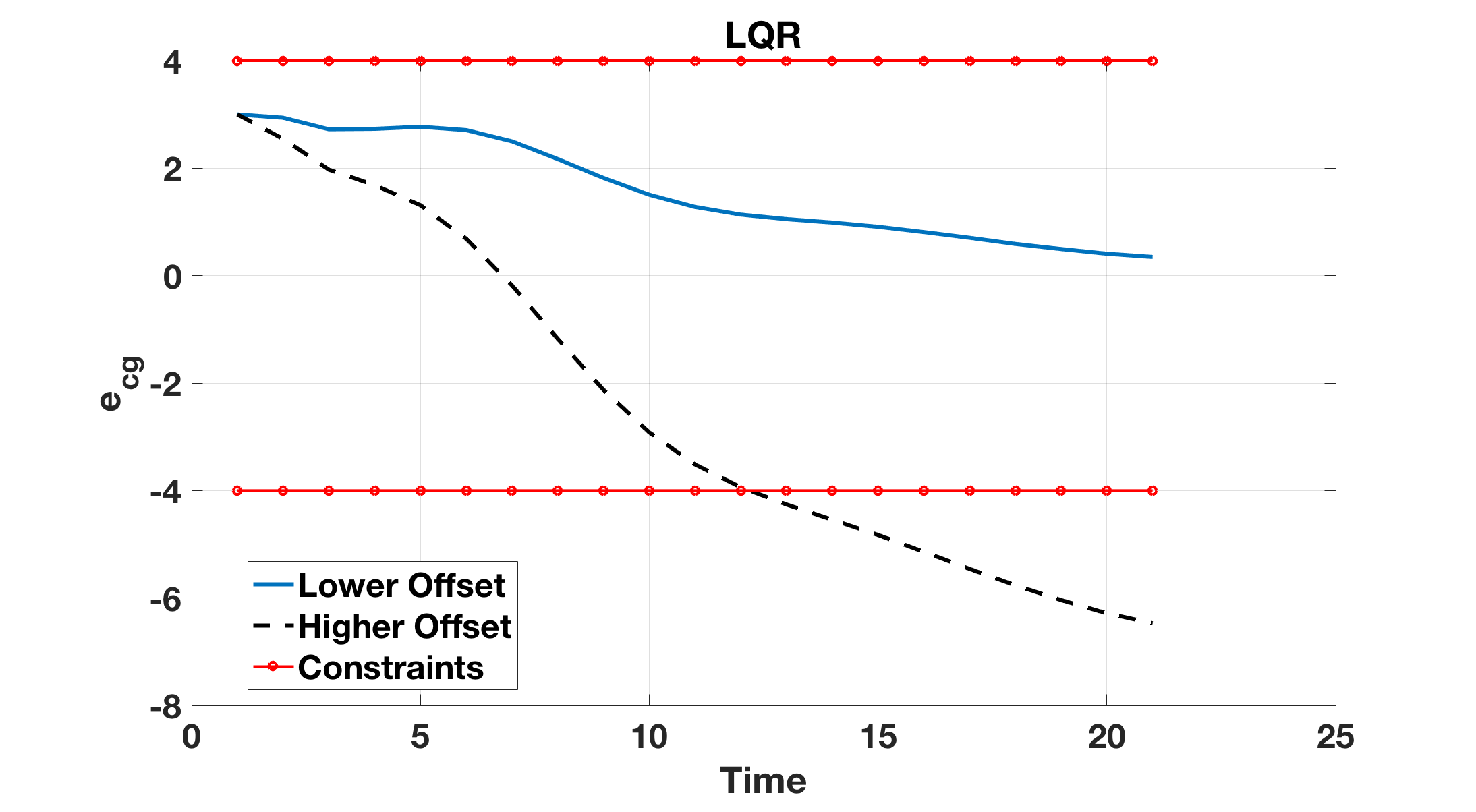}
	\caption{LQR for Varying Offset}
	\label{fig:reset_c_lqr}
\end{figure}

\section{CONCLUSIONS AND FUTURE WORK}\label{sec:concl}
In this paper, we developed an Adaptive Robust MPC strategy for lane-keeping of vehicles in presence of a steering angle offset. We showed that with every time step, the knowledge of the offset is learned and thus improved. The algorithm is robust against the true steering angle offset and satisfies imposed constraints recursively. Thus, the Adaptive Robust MPC is shown to have improved performance over a standard MPC, which is not cognizant of the model mismatch. We also demonstrated that our algorithm is easy to reset and implement without re-tuning from the scratch, if any variation in the uncertainty is suspected over time. Such aspects can be useful in works such as \cite{bujarbaruahlyap}. In future extensions of this work, we aim to solve a Robust MPC problem with time varying uncertainty. 

%%%%%% BIB HERE. FROM THE TEMPLATE 

%=======================================================================
\printbibliography

\appendix
\section{Appendix}
\subsection{Matrix Definitions}
The matrices pertaining to the lateral dynamics of the vehicle, $A^\textnormal{c}, B^\textnormal{c}_1$ and $B^\textnormal{c}_2$ in (\ref{eq:car_M_brief}) are defined as,
\begin{align*}
    A^\textnormal{c} & = \begin{bmatrix} 0 & 1 & 0 & 0\\ 0 & -\frac{c_{\textnormal{f}}+c_{\textnormal{r}}}{MV_x} & \frac{c_{\textnormal{f}}+c_{\textnormal{r}}}{M} & \frac{l_{\textnormal{r}}c_{\textnormal{r}}-l_{\textnormal{f}}c_{\textnormal{f}}}{MV_x}\\ 0 & 0 & 0 & 1\\ 0 & \frac{l_{\textnormal{r}}c_{\textnormal{r}}-l_{\textnormal{f}}c_{\textnormal{f}}}{J_zV_x}& \frac{l_{\textnormal{f}}c_{\textnormal{f}}-l_{\textnormal{r}}c_{\textnormal{r}}}{J_z} & -\frac{l_{\textnormal{f}}^2c_{\textnormal{f}}+l_{\textnormal{r}}^2c_{\textnormal{r}}}{J_zV_x} \end{bmatrix},\\
    B^\textnormal{c}_1 & = [0, \frac{c_{\textnormal{f}}}{M}, 0, \frac{l_{\textnormal{f}}c_{\textnormal{f}}}{J_z}]^\top,\\
    B_2^\textnormal{c} & = [0, \frac{l_{\textnormal{r}}c_{\textnormal{r}}-l_{\textnormal{f}}c_{\textnormal{f}}}{MV_x}-V_x, 0, -\frac{l_{\textnormal{f}}^2c_{\textnormal{f}} + l_{\textnormal{r}}^2c_{\textnormal{r}}}{J_zV_x}]^\top,
\end{align*}
where $M$ is the mass, $J_z$ is the moment of inertia about the vertical axis, $l_{\textnormal{f}},l_{\textnormal{r}}$ are distances from center of gravity to front and rear axles respectively, and $c_{\textnormal{f}},c_{\textnormal{r}}$ are cornering stiffnesses of front and rear tires respectively, of the vehicle considered. $V_x$ is the vehicle longitudinal speed.

\subsection{MPC Reformulation}
In this section we show how the robust MPC problem (\ref{eq:rob_MPC}) can be reformulated and efficiently solved. The constraints in (\ref{eq:rob_MPC}) can be compactly written with similar notations as \cite{Goulart2006},
\begin{equation}\label{eq:com_con_appen}
    F_t\mathbf{v_t} + G_t(\mathbf{w}_t + \mathbf{E}\pmb{\theta_t}) \leq c_t + H_t \Delta x_{t},
\end{equation}
where we denote, $\mathbf{v_t} = [v_{t|t}^\top, v_{t+1|t}^\top, \cdots, v_{t+N-1|t}^\top]^\top \in \mathbb{R}^{mN}$, $\pmb{\theta_t} = [\theta_{t|t}^\top, \cdots, \theta_{t+N-1|t}^\top]^\top\in \mathbb{R}^{pN}$ for all $\theta_{k|t} \in \Theta_t,~\forall k\in [t,\cdot,t+N-1]$,  $\mathbf{E} = \textnormal{diag}(E,\cdots,E) \in \mathbb{R}^{nN \times pN}$ and $\mathbf{w_t} = [w_{t|t}^\top, \cdots, w_{t+N-1|t}^\top]^\top \in \mathbb{R}^{nN}$. 

The matrices above in (\ref{eq:com_con_appen}) $F_t \in \mathbb{R}^{(sN+r_t)\times mN}, G_t \in \mathbb{R}^{(sN+r_t)\times nN}, c_t \in \mathbb{R}^{sN+r_t}$ and $H_t \in \mathbb{R}^{(sN+r_t)\times n}$ are formed after expressing all states and constraints in terms of the initial condition $\Delta x_t$ before the finite horizon problem (\ref{eq:rob_MPC}) is solved. They can be obtained as,
\begin{align*}
&F_t = \begin{bmatrix}D&\mathbf{0}&\cdots&\mathbf{0}\\\bar{C}B_1 & D & \cdots & \mathbf{0} \\ \vdots & \ddots & \ddots & \vdots \\ \bar{C}A_{\textnormal{cl}}^{N-2}B_1 & \bar{C}A_{\textnormal{cl}}^{N-3}B_1 & \cdots & D\\Y_tA_{\textnormal{cl}}^{N-1}B_1 & Y_tA_{\textnormal{cl}}^{N-2}B_1 & \cdots & Y_tB_1 \end{bmatrix}\\
&G_t=\begin{bmatrix} \mathbf{0}&\mathbf{0}&\cdots & \mathbf{0}\\\bar{C}&\mathbf{0}&\cdots&\mathbf{0}\\ \vdots & \ddots & \ddots & \vdots \\ \bar{C}A_{\textnormal{cl}}^{N-2}&\bar{C}A_{\textnormal{cl}}^{N-3}&\cdots&\mathbf{0}\\Y_tA_{\textnormal{cl}}^{N-1}&Y_tA_{\textnormal{cl}}^{N-2}&\cdots&Y_t \end{bmatrix},\\
&c_t = [b^\top,\cdots,b^\top,z_t^\top],\\
&H_t=-[\bar{C}^\top,(\bar{C}A_{\textnormal{cl}})^\top,\cdots,(\bar{C}A_{\textnormal{cl}}^{N-1})^\top,(Y_tA_{\textnormal{cl}}^N)^\top]^\top,
\end{align*}
where $\bar{C}=C-DK$. So, while attempting to solve (\ref{eq:com_con_appen}) at time $t$ for robust constraint satisfaction, we must have,
\begin{equation}\label{eq:com_con_max_appen}
    F_t\mathbf{v_t} +  \max_{\mathbf{w_t},\pmb{\theta_t}} G_t(\mathbf{w}_t + \mathbf{E}\pmb{\theta_t}) \leq c_t + H_t \Delta x_{t}.
\end{equation}
Therefore, the above equation (\ref{eq:com_con_max_appen}) indicates that (\ref{eq:rob_MPC}) can be solved by imposing constraints (\ref{eq:constraints}) on nominal states ($\Delta 
\bar x_{k|t}$) after tightening them throughout the horizon (of length $N$) by $\max_{\mathbf{w_t},\pmb{\theta_t}}(\mathbf{w}_t + \mathbf{E}\pmb{\theta_t})$ \cite[sec.~3.2]{kouvaritakis2016model}. Alternatively, denote the polytope $\mathbb{S}_{t}=\{w \in \mathbb{W},~\theta \in \Theta_t: S_t(\mathbf{w}+ \mathbf{E}\pmb{\theta}) \leq h_{t},~S_{t} \in \mathbb{R}^{a_t \times nN},~ h_t \in \mathbb{R}^{a_t}\}$, then (\ref{eq:com_con_max_appen}) can be written with auxiliary decision variables $Z_t \in \mathbb{R}^{a_t \times (sN+r_t)}$ using the concept of  \emph{duality} of linear programs as,
\begin{align}\label{eq:com_con_dual}
& F_t\mathbf{v_t} + Z_t^\top h_t \leq c_t + H_t \Delta x_t,~G_t = Z_t^\top S_t,~Z_t  \geq 0,
\end{align}
which is a tractable linear programming problem that can be efficiently solved with any existing solver for real time implementation of the algorithm.

\end{document}